\RequirePackage{amsmath}
\documentclass[]{llncs}
\usepackage{amsmath,amssymb,mathtools}
\usepackage[linesnumbered, ruled]{algorithm2e}
\usepackage[caption=false]{subfig}
\usepackage[T1]{fontenc}
\usepackage[inline]{enumitem}

\DeclareMathOperator*{\argmax}{arg\,max}

\def\R{\mathbb{R}}
\def\N{\mathbb{N}}

\newtheorem{observation}{Observation}

\newcommand{\din}[1][M]{deg^M_{in}}
\newcommand{\dout}[1][M]{deg^M_{out}}

\def\UWCARPOOL{Unweighted Carpool Matching}
\def\CARPOOL{Maximum Carpool Ma\-tching}
\def\FIXEDCARPOOL{Fixed Maximum Carpool Matching}

\SetKw{False}{false}
\usepackage{tikz}
\usetikzlibrary{
	arrows,
	calc,
	graphs,
	patterns,
	positioning,
	shapes,
	decorations.pathmorphing,
}

\tikzset{default style/.style={
	on grid, 
	auto, 
	thick,
}}

\tikzset{default node/.style={
	draw, 
	circle,
	inner sep=0mm,
	minimum size=5mm,
	very thick,
	font=\small,
	black!70,
}}

\tikzset{label/.style={
	draw=none,
	sloped,
}}

\tikzset{label above/.style={
	label,
	midway,
	above=-1mm,
}}

\tikzset{label below/.style={
	label,
	midway,
	below=-1mm,
}}

\title{Approximation Algorithms for the Maximum Carpool Matching Problem}

\author{
Reuven Bar-Yehuda\inst{1}
\and
Gilad Kutiel\inst{1}
\and
Dror Rawitz\inst{2}
}

\institute{
Department of Computer Science, Technion, Haifa, Israel
\\
\email{\{reuven, gkutiel\}@cs.technion.ac.il}
\and
Faculty of Engineering, Bar Ilan University, Ramt-Gan, Israel
\\
\email{dror.rawitz@biu.ac.il}
}

\date{}
\author{
Gilad Kutiel\inst{1}
}

\institute{
Department of Computer Science, Technion, Haifa, Israel
\\
\email{gkutiel@cs.technion.ac.il}
}

\begin{document}
\maketitle

\begin{abstract}
The \textsc{\CARPOOL{}} problem is a star packing problem in directed graphs.
Formally, given a directed graph $G = (V, A)$,
a capacity function $ c: V \rightarrow \N $,
and a weight function $w : A \rightarrow \R $,
a feasible \emph{carpool matching} is a triple 
$(P, D, M)$, where $P$ (passengers) and $D$ (drivers) form a partition of $V$, 
and $M$ is a subset of $A \cap (P \times D)$,
under the constraints that for every vertex $d \in D$, 
$\din(d) \leq c(d)$, 
and for every vertex $p \in P$, $\dout(p) \leq 1$.
In the \textsc{\CARPOOL{}} problem we seek for a matching $(P, D, M)$ that maximizes the
total weight of $M$.

The problem arises when designing an online carpool service, 
such as Zimride~\cite{zimride}, 
that tries to connect between passengers and drivers based on (arbitrary) similarity function.
The problem is known to be NP-hard, 
even for uniform weights and without capacity constraints.

We present a $3$-approximation algorithm for the problem
and $2$-approx\-imation algorithm for the unweighted variant of the problem.
\end{abstract}

\section{Introduction}
Carpooling, is the sharing of car journeys so that more than one person travels
in a car.
Knapen et al.~\cite{knapen2013estimating} describe an automatic service
to match commuting trips.
Users of the service register their personal profile and a set of periodically
recurring trips, 
and the service advises registered candidates on how to combine their commuting
trips by carpooling.
The service acts in two phases. 

In the first phase, the service estimates the probability that a person $a$
traveling in person's $b$ car will be satisfied by the trip.
This is done based on personal information and feedback from users on past
rides.
The second phase is about finding a carpool matching
that maximizes the global (total expected) satisfaction.

The second phase can be modeled in terms of graph theory.
Given a directed graph $G = (V, A)$.
Each vertex $v \in V$ corresponds to a user of the service and an arc
$(u, v)$ exists if the user corresponding to vertex $u$ is willing to
commute with the user corresponding to vertex $v$.
A capacity function $ c: V \rightarrow \N $ is defined
according to the number of passengers each user can drive if she was
selected as a driver.
A weight function $w : A \rightarrow \R $ defines the amount of
satisfaction $w(u, v)$,
that user $u$ gains when riding with user $v$.

A feasible \emph{carpool matching} (matching) is a triple 
$(P, D, M)$, where $P$ and $D$ form a partition of $V$, 
and $M$ is a subset of $A \cap (P \times D)$,
under the constraints that for every driver $d \in D$, 
$\din(d) \leq c(d)$, 
and for every passenger $p \in P$, ${\dout(p) \leq 1}$.
In the \textsc{\CARPOOL{}} problem we seek for a matching $(P, D, M)$ that maximizes the
total weight of $M$.
In other words, the \textsc{\CARPOOL{}} problem is about finding a set of 
(directed toward the center) vertex disjoint stars 
that maximizes the total weights on the arcs.
Figure~\ref{fig:carpool} is an example of the \textsc{\CARPOOL{}} problem.

\begin{figure}
\centering
\newcommand{\arc}[3]{
	\draw (#1) -- (#2) node[label above] {#3};
}
\subfloat[]{
\label{subfloat:input}
\begin{tikzpicture}[every node/.style={default node}, ->, very thick]

\node(1) at(-1,3) {2};
\node(2) at(.5,.5) {3};
\node(3) at(-1,0) {1};
\node(4) at(-1,1) {3};
\node(5) at(-2,2) {0};
\node(6) at(-2,-1) {4};
\node(7) at(-3,0) {3};
\node(8) at(.5,-1) {2};
\node(9) at(0,2) {1};
\node(10) at(2,1) {4};

\arc{1}{4}{3}
\arc{2}{4}{4}
\arc{3}{4}{5}

\arc{5}{7}{2}
\arc{6}{7}{4}

\arc{8}{10}{6}
\arc{9}{10}{2}

\arc{1}{5}{4}
\arc{4}{5}{4}

\arc{2}{3}{2}
\arc{6}{3}{2}
\arc{7}{3}{2}

\arc{8}{2}{1}
\arc{10}{2}{1}

\arc{3}{8}{3}
\arc{8}{3}{3}

\arc{1}{9}{1}
\arc{4}{9}{1}

\end{tikzpicture}}
\hfill\subfloat[]{
\label{subfloat:output}
\begin{tikzpicture}[every node/.style={default node}, ->, very thick]

\begin{scope}[every node/.style={default node, draw=blue}]
\node(5) at(-2,2) {0};
\node(6) at(-2,-1) {4};

\node(2) at(.5,.5) {3};
\node(3) at(-1,0) {1};
\node(1) at(-1,3) {2};

\node(8) at(.5,-1) {2};
\node(9) at(0,2) {1};
\end{scope}

\begin{scope}[every node/.style={default node, draw=red, dashed}]
\node(7) at(-3,0) {3};

\node(4) at(-1,1) {3};

\node(10) at(2,1) {4};
\end{scope}

\arc{1}{4}{3}
\arc{2}{4}{4}
\arc{3}{4}{5}

\arc{5}{7}{2}
\arc{6}{7}{4}

\arc{8}{10}{6}
\arc{9}{10}{2}

\end{tikzpicture}}
\caption[]{
\label{fig:carpool}
A carpool matching example: 
\subref{subfloat:input} 
a directed graph with capacities on the vertices and weights on the arcs. 
\subref{subfloat:output}
a feasible matching with total weight of 26.
$P$ is the set of blue vertices, and $D$ is the set of red, dashed vertices. 
}
\end{figure}  

Hartman et al.~\cite{hartman2013optimal} proved that the
\emph{\CARPOOL{}} problem considered in this paper is NP-hard,
and that the problem remains NP-hard even for a binary weight function when
the capacity function $c(v) \leq 2$ for every vertex in $V$.
It is also worth mentioning, that in the undirected, uncapacitated, unweighted
variant of the problem, the set of drivers in an optimal solution
form a minimum dominating set.
When the set of drivers is known in advanced, however, the problem becomes
tractable and can be solved using a reduction to a flow network problem.

Agatz et al.~\cite{agatz2012optimization} outlined the optimization challenges
that arise when developing technology to support ride-sharing and survey the
related operations research models in the academic literature.  
Hartman et al.~\cite{hartman2014theory} designed several heuristic algorithms
for the \CARPOOL{} problem and compared 
their performance on real data.
Other heuristic algorithms were developed as well~\cite{knapen2014exploiting}.
Arkin et al.~\cite{arkin2004approximations}, considered other variants of
capacitated star packing where a capacity vector is given as part of the input and 
capacities need to be assigned to vertices.  

Nguyen et al.~\cite{nguyen2008approximating} considered the \textsc{spanning star forest} problem
(the undirected, uncapacitated, unweighted variant of the problem).
They proved the following results:
\begin{enumerate*}
\item
there is a polynomial-time approximation scheme for planner graphs;
\item 
there is a polynomial-time $\frac{3}{5}$-approximation algorithm for graphs;
\item 
there is a polynomial-time $\frac{1}{2}$-approximation algorithm for weighted graphs.
\end{enumerate*}
They also showed how to apply the spanning star forest model to aligning multiple
genomic sequences over a tandem duplication region.
Chen et al.~\cite{chen2007improved} improved the approximation ratio to 0.71,
and also showed that the problem can not be approximated to within a factor of
$\frac{31}{32} + \epsilon$ for any $\epsilon > 0$ under the assumption 
that $\text{P} \neq \text{NP}$.
It is not clear, however, if any of the technique used to address the
\textsc{spanning star forest} problem can be generalized to approximate the
directed capacitated variant.

In section~\ref{sec:fixed} we present an exact, efficient algorithm for the
problem when the set of drivers and passengers is given in advanced.
In section~\ref{sec:uwcm} we present a 2-approximation local search algorithm
for the unweighted variant of the problem.
Finally in section~\ref{sec:cm} we
give a 3-approximation algorithm for the problem.

\section{Maximum Weight Flow}
\label{sec:carpool:preliminary}
A flow network is a tuple $N = (G = (V, A), s, t, c)$, 
Where $G$ is a directed graph, 
$s \in V$ is a source vertex, 
$t \in V$ is a target vertex, 
and $c : A \rightarrow \R$ is a capacity function. 
A flow $f : A \rightarrow \R$ is a function that has the following properties:
\begin{itemize}
\item
$f(e) \leq c(e), \quad \forall e \in A$

\item
$\sum_{(u, v) \in A} f(u, v) = \sum_{(v, w) \in A} f(v, w), \quad \forall v \in V \setminus \{s, t\}$
\end{itemize}

Given a flow function $f$, 
and a weight function $w: A \rightarrow \R$, 
the \emph{flow weight} is defined to be:
$\sum_{e \in A}{w(e)f(e)}$.
A flow with a maximum weight (\emph{maximum weight flow}) can be efficiently found by adding 
the arc $(t, s)$, with $c(t,s) = \infty$, and $w(t,s) = 0$ and reducing the problem
(by switching the sign of the weights) 
to the minimum cost circulation problem~\cite{tardos1985strongly}.
When the capacity function $c$ is integral, 
a maximum weight integral flow can be efficiently found.

\section{\FIXEDCARPOOL{}}
\label{sec:fixed}
In the \textsc{\FIXEDCARPOOL{}} problem, $P$ and $D$ are given, 
and the goal is to find $M$ that maximizes the total weight. 
This variant of the problem can be solved efficiently
\footnote{A solution to this variant of the problem was already proposed in~\cite{hartman2014theory}.
For the sake of completeness, however, we describe a detailed solution for this variant. 
More importantly, 
the described solution helps us develop the intuition and understand the basic idea behind the
approximation algorithm described in Section~\ref{sec:cm}.   
},
by reducing it to a maximum weight flow (flow) problem as
follow:
Let $(G = (V, A), c, w)$ be a \CARPOOL{} instance,
let $(P, D)$ be a partition of $V$,
let  $N = (G' = (V', A'), s, t, c')$ be a flow network, 
and let $w' : A \rightarrow \N$ be a weight function, where 

\begin{align*}
V'			& = P \cup D \cup \{s, t\}										\\
A'			& = A_{sp} \cup A_{pd} \cup A_{dt}								\\
A_{sp}		& =	\{(s, p) : p \in P \}										\\
A_{pd}		& =	A \cap (P \times D)											\\
A_{dt}		& =	\{(d, t) : d \in D \}										\\
c'(u, v)	& = 
				\begin{cases}
				c(u) & \text{if } (u, v) \in A_{dt} 						\\
				1 & \text{otherwise}
				\end{cases}
																			\\
w'(e)			& = 
				\begin{cases}
				w(e) & \text{if } e \in A_{pd} 								\\
				0 & \text{otherwise}	
				\end{cases}
\end{align*}
The flow network is described in Figure~\ref{fig:flow}.
\begin{figure}
\centering
\begin{tikzpicture}[every node/.style={default node}]

\node(s) at (0,0) {s};

\node[draw=none]() 		at (3,2.2) {$P$};
\node(p0) 				at (3,1.6) {$p_0$};
\node[draw=none](pdots1)at (3,.8) {$\vdots$};
\node(pi) 				at (3,0) {$p_i$};
\node[draw=none](pdots2)at (3,-.8) {$\vdots$};
\node(pl) 				at (3,-1.6) {$p_l$};

\node[draw=none]() 		at (7,2.2) {$D$};
\node(d0) 				at (7,1.6) {$d_0$};
\node[draw=none](pdots)	at (7,.8) {$\vdots$};
\node(dj) 				at (7,0) {$d_j$};
\node[draw=none](pdots)	at (7,-.8) {$\vdots$};
\node(dm) 				at (7,-1.6) {$d_m$};

\node(t) at (10,0) {t};

\draw[->] (s) -- (pi)
node[label above] {$w' = 0$}
node[label below] {$c' = 1$};

\draw[->] (pi) -- (dj)
node[label above, above=-5mm] {$w' = w(p_i, d_j)$}
node[label below] {$c' = 1$};

\draw[->] (dj) -- (t)
node[label above] {$w' = 0$}
node[label below, below=-3mm] {$c' = c(d_j)$};

\newcommand{\edots}[2]{
\path (#1) -- (#2)
node[label, pos=0.1, anchor=center] {$\cdots$}
node[label, pos=0.9, anchor=center] {$\cdots$};
}

\edots{s}{p0}
\edots{s}{pl}

\edots{p0}{d0}
\edots{p0}{dj}
\edots{pi}{d0}
\edots{pi}{dm}
\edots{pl}{dj}
\edots{pl}{dm}

\edots{d0}{t}
\edots{dm}{t}

\end{tikzpicture}
\caption{
\label{fig:flow}
Illustration of a flow network corresponding to a \FIXEDCARPOOL{} instance.}
\end{figure}

\begin{observation}
For every integral flow $f$ in $N$, there is a carpool matching $M$ on $G$ with
the same weight.
\end{observation}

\begin{proof}
Consider the carpool matching $(P, D, M^f)$, where 
$$ M^f = \{(p, d) \in A_{pd} : f(p, d) = 1\} $$
one can verify that this is indeed a matching with the same weight as $f$.
\end{proof}

\begin{observation}
For every carpool matching $(P, D, M)$ on $G$, there exists a flow $f$ on $N$
with the same weight.
\end{observation}

\begin{proof}
Consider the flow function
\begin{align*}
f(s, p_i)		& = \dout(p_i)		 				\\
f(p_i, d_j)		& = 
				\begin{cases}
				1 & \text{if } (p_i, d_j) \in M		\\
				0 & \text{otherwise}
				\end{cases}						\\
f(d_j, t) 	& = \din(d_j) 
\end{align*}

It is easy to verify, that $f$ is indeed a flow function.
Also, observe, that by construction,
the weight of $f$ equals to the weight of the matching.
\end{proof}

As we mentioned, 
the maximum weight flow problem can be solved efficiently, 
and so is the \FIXEDCARPOOL{} problem.
It is worth mentioning, that it is possible that in a maximum weight flow, 
some of the arcs will have no flow at all, 
that is, it is possible that in a \FIXEDCARPOOL{}
some of the passengers and drivers will be unmatched.  

\section{\UWCARPOOL{}}
\label{sec:uwcm}
In this section we present a local search algorithm for the unweighted
variant of the problem.
We show that the approximation ratio of this algorithm is $2$ and give an example
to show that our analysis is tight.

Given a directed graph $G = (V, A)$, 
and a capacity function ${c : V \rightarrow \N}$, 
In the \textsc{\UWCARPOOL{}} problem, 
we seek for a matching that maximizes the size of $M$.

We now present a simple local search algorithm for the problem. 
The algorithm maintains a feasible matching through its execution.
In every iteration of the algorithm, the size of $M$ increases.
The algorithm terminates, when no further improvement can be made. 

Recall that the \FIXEDCARPOOL{} can be solved efficiently.
Let $M = \text{opt}_{fixed}(P, D)$ be an optimal solution of the
\FIXEDCARPOOL{} problem.
For a given matching $M$, define the following sets:
\begin{itemize}
\item $P^M = \{v : \dout(v) = 1\}$
\item $D^M = \{v : \din(v) > 0\}$
\item $D^M_c = \{v : \din(v) = c(v)\}$
\item $F^M = \{v : \din(v) = \dout(v) = 0\}$ 
\end{itemize}
We refer to the vertices in these sets as, \emph{passenger}, 
\emph{driver}, \emph{saturated driver}, and \emph{free vertex} respectively.
The local search algorithm, in every iteration, 
tries to improve the current matching, 
by switching a passenger or a free vertex into a driver 
and compute an optimal fixed matching.
The local search algorithm is described in
Algorithm~\ref{alg:local}.

\begin{algorithm}
\SetKw{True}{true}
\SetKw{False}{false}
\KwIn{$G = (V, A)$, $c : V \rightarrow \N$}
\KwOut{$M$}
$M \leftarrow \emptyset$					\\
\Repeat{done}{
\label{line:outerloop}
 	$done \leftarrow{}$ \True				\\
 	\For{$v \in (V \setminus D^M)$}{
 		$D \leftarrow D^M \cup \{v\}$ 		\\
 		$P \leftarrow V \setminus D$ 		\\
 		$M' = \text{opt}_{fixed}(P, D)$ 	\\
		\If{$|M'| > |M|$}{
			$M \leftarrow M'$				\\
			$done \leftarrow{}$ \False		\\
		}
 	}
}
\Return{$M$}

\caption{
\label{alg:local}
Local Search}
\end{algorithm}

First, observe that the outer loop on line~\ref{line:outerloop} of the local search algorithm
can be executed at most $n$ times, 
where $n$ is the total number of vertices, 
this is because the loop is executed only when there was an improvement, 
and this can happen at most $n$ times.
Also, observe that the body of this loop can be computed in polynomial time, 
and we can conclude that Algorithm~\ref{alg:local} runs in polynomial time.
    
We now prove that the local search algorithm achieves an approximation ratio of $2$.
Let $M$ be a matching found by the local search algorithm, 
and let $M^*$ be an arbitrary but fixed optimal matching.
Observe that every arc in $M^*$ has at least one end point in $M$, formally: 
\begin{observation}
If $(u,v) \in M^*$, then $\{u,v\} \cap (P^M \cup D^M) \neq \emptyset$
\end{observation}

\begin{proof}
If this is not the case, Algorithm~\ref{alg:local} can improve $M$
by adding the arc $(u,v)$.  
\end{proof}

Now, with respect to $M$, the optimal solution can not match two free vertices to the same
passenger, formally:
\begin{observation}
\label{observation:one-free}
If $(p,d) \in M$, $f_1, f_2 \in F^M$, and $(f_1, p), (f_2, p) \in M^*$, then $f_1 = f_2$.
\end{observation}

\begin{proof}
If this is not the case, Algorithm~\ref{alg:local} can improve $M$ 
by removing the arc $(p,d)$ and adding the arcs $(f_1, p), (f_2, p)$.
\end{proof}

Finally, with respect to $M$, the optimal solution can not match a free vertex to a driver
that is not saturated, formally: 
\begin{observation}
\label{observation:saturated}
If $(f,d) \in M^*$, $f \in F^M$, and $d \in D^M$, then $d \in D^M_c$.
\end{observation}

\begin{proof}
If this is not the case, once again, Algorithm~\ref{alg:local} can improve $M$
by adding the arc $(f,d)$.
\end{proof}

To show that Algorithm~\ref{alg:local} is 2-approximation, 
consider the charging scheme that is illustrated in Figure~\ref{fig:charging}:
Load every arc $(p,d) \in M$ with 2 coins, 
place one coin on $p$ and one coin on $d$.
Observe that every vertex $p \in P^M$ is loaded with one coin, 
and every vertex $d \in D^M$ is loaded with $\din(d)$ coins.   
Now, pay one coin for every $(u,v) \in M^*$, charge $u$ if $u \in P^M \cup D^M$, 
otherwise ($v \in P^M \cup D^M$) charge $v$.
Clearly, every arc in $M^*$ is paid.
We claim that no vertex is overcharged.

\begin{figure}
\centering
\begin{tikzpicture}[
very thick, 
->, 
every pin arc/.style={
	-, 
	very thin,
	decorate,
	decoration={snake},
	black!60,
}]

\tikzset{
	opt/.style={blue, dashed},
	d/.style={default node, draw=red, -},
	p/.style={default node, draw=orange, dashed,-},
	f/.style={default node, draw=gray, dotted,-},
}

\draw[] (3,1.3) -- (3.4,1.3) node[right]{\tiny $M$};
\draw[opt] (3,1) -- (3.4,1) node[right]{\tiny $M^*$};
\draw[p] (3,.7) -- (3.4,.7) node[right]{\tiny $P^M$};
\draw[d] (3,.4) -- (3.4,.4) node[right]{\tiny $D^M$};
\draw[f] (3,.1) -- (3.4,.1) node[right]{\tiny $F^M$};

\begin{scope}[every node/.style={d}]
\node(3)[pin=above right:{\small 2\$}] at(0,.5)			{3};
\node(7)[pin=below right:{\small 2\$}] at(-1.5, -1)		{7};
\end{scope}

\begin{scope}[every node/.style={p}]
\node(1)[pin=above left:{\small 1\$}] 	at(-1,1)		{1};
\node(2)[pin=right:{\small 1\$}] 		at(.8,-.2)		{2};
\node(4)[pin=left:{\small 1\$}] 		at(-2,0)		{4};
\node(5)[pin=right:{\small 1\$}] 		at(-.4,-1.2)	{5};
\end{scope}

\begin{scope}[every node/.style={f}]
\node(8)								at(-2.8,-1.6)	{8};
\node(9)								at(-2,-2.2)		{9};
\node(10)								at(.3,-2.2)		{10};
\end{scope}

\begin{scope}
\draw (1) to (3);
\draw (2) to (3);
\draw (4) to (7);
\draw (5) to (7);
\end{scope}

\begin{scope}[opt]
\draw (1) to (4);
\draw (8) to (7);
\draw (9) to (7);
\draw (10) to (5);
\draw (2) to (5);
\draw (3) to (5);
\end{scope}

\end{tikzpicture}
\caption[]{
\label{fig:charging}
Charging Scheme:																	\\
1. vertices $1,2,4,5$ are loaded with 1\$ each and vertices $3,7$ with 2\$ each.	\\
2. vertex 1 pays for the arc $(1,4)$.												\\
3. vertex 5 pays for the arc $(10,5)$.												\\
4. vertex 7 is saturated. It pays for arcs $(8,7)$ and (9,7). 
}
\end{figure}
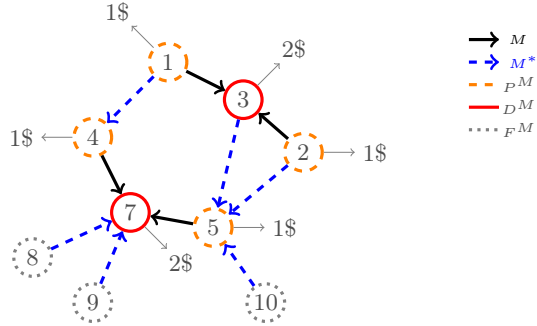

\begin{observation}
\label{observation:p-not-charged}
If $u \in P^M$, then $u$ is not overcharged.
\end{observation}

\begin{proof}
If $u \in P^{M^*}$, then it is only charged once, otherwise, 
if $u \in D^{M^*}$, then it is only charged for arcs $(w, u)$ where $w \in F^M$,
and by Observation~\ref{observation:one-free}, there is at most one such arc. 
\end{proof}

\begin{observation}
\label{observation:d-not-charged}
If $u \in D^M$, then $u$ is not overcharged.
\end{observation}

\begin{proof}
If $u \in P^{M^*}$, then it is only charged once, 
if $u \in D^{M^*}$, then it is only charged for arcs $(w, u)$ where $w \in F^M$,
if such arcs exists, then by observation~\ref{observation:saturated}, $u$ is saturated, 
and can not be overcharged.
\end{proof}

\begin{theorem}
Algorithm~\ref{alg:local} is 2-approximation
\end{theorem}

\begin{proof}
We use a charging scheme where we manage to pay 1 coin for each arc in $M^*$
by using at most $2|M|$ coins.
\end{proof}

To conclude this section, we show that our analysis is tight.
Consider the example given in Figure~\ref{fig:localtight}.
Assume, in this example, that there are no capacity constraints,
if the local search algorithm starts by choosing vertex $3$ to be a driver, 
then the returned matching is the single arc $(2,3)$.
At this point, no further improvement can be done.
The optimal matching, on the other hand, is $\{(1, 2), (3, 2)\}$. 
The path in the example can be duplicated to form an arbitrary large graph (forest).

\begin{figure} 
\centering
\begin{tikzpicture}[every node/.style={default node}]

\node(a) at (0,0) {$1$};
\node(b) at (2,0) {$2$};
\node(c) at (4,0) {$3$};

\draw[->, dashed, very thick] (a) -- (b);
\draw[<->, very thick](b) -- (c);

\end{tikzpicture}
\caption{
\label{fig:localtight}
Local Search - Worst Case Example
}
\end{figure}
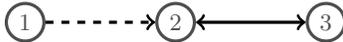

\section{\CARPOOL{}}
\label{sec:cm}
	\subsection{Super Matching}
A super-matching is a relaxed variant of the \CARPOOL{} problem 
where every node can act both as a driver and as a passenger. 
Formally, given a directed graph $G = (V, A)$, 
a capacity function $ c: V \rightarrow \N $,
and a weight function $w : A \rightarrow \R$,
a \emph{super-matching} is a set $M \subseteq A$, 
under the constraint that $\forall v \in V$,
$\din(v) \leq c(v)$, and $\dout(v) \leq 1$. 
Clearly, the following observation holds:
\begin{observation}
\label{ob:matching is super}
Every matching $(P, V, M)$ is a super-matching $M$
\end{observation}

A maximum super matching can be found efficiently by the following reduction 
to a maximum weight flow problem: 
Let $N = (G', s, t, c', w')$ be a flow network, where 
\begin{align*}
G'				& = (P \cup D \cup \{s, t\}, A_{sp} \cup A_{pd} \cup A_{dt})	\\
P				& = \{p_v : v \in V\}					\\
D				& = \{d_v : v \in V\}					\\
A_{sp}			& = \{ (s, p_v) : p_v \in P \}			\\
A_{pd}			& = \{ ((p_u, d_v)) : (u, v) \in A \}	\\
A_{dt}			& = \{ (d_v, t) : d_v \in D \}			\\
c'(s, p_v)		& = c'(p_u, d_v) = 1					\\
c'(d_v, t)		& = c(v)								\\
w'(p_u, d_v)	& = 
\begin{cases}
w(u, v) & \text{if } (p_u, d_v) \in A_{pd} \\
0 & \text{otherwise}	
\end{cases}
\end{align*}

That is, 
we construct a bipartite graph where the left side represents each vertex 
in $V$ being a passenger,
and the right side represents each vertex in $V$ being a driver.
Figure~\ref{fig:cm-flow} illustrates this flow network.
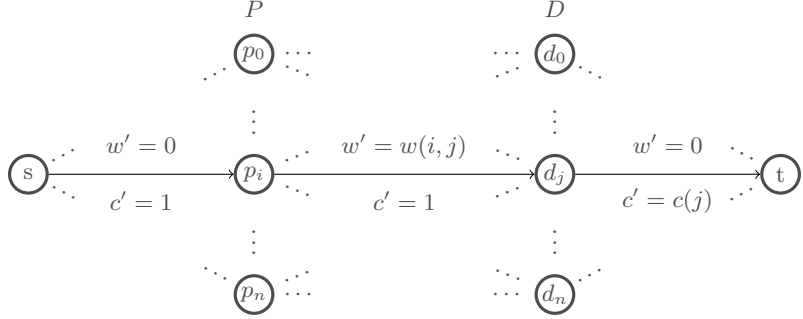
\begin{figure}
\centering
\begin{tikzpicture}[every node/.style={default node}]

\node(s) at (0,0) {s};

\node[draw=none]() 		at (3,2.2) {$P$};
\node(p0) 				at (3,1.6) {$p_0$};
\node[draw=none](pdots1)at (3,.8) {$\vdots$};
\node(pi) 				at (3,0) {$p_i$};
\node[draw=none](pdots2)at (3,-.8) {$\vdots$};
\node(pl) 				at (3,-1.6) {$p_n$};

\node[draw=none]() 		at (7,2.2) {$D$};
\node(d0) 				at (7,1.6) {$d_0$};
\node[draw=none](pdots)	at (7,.8) {$\vdots$};
\node(dj) 				at (7,0) {$d_j$};
\node[draw=none](pdots)	at (7,-.8) {$\vdots$};
\node(dm) 				at (7,-1.6) {$d_n$};

\node(t) at (10,0) {t};

\draw[->] (s) -- (pi)
node[label above] {$w' = 0$}
node[label below] {$c' = 1$};

\draw[->] (pi) -- (dj)
node[label above, above=-5mm] {$w' = w(i, j)$}
node[label below] {$c' = 1$};

\draw[->] (dj) -- (t)
node[label above] {$w' = 0$}
node[label below, below=-3mm] {$c' = c(j)$};

\newcommand{\edots}[2]{
\path (#1) -- (#2)
node[label, pos=0.1, anchor=center] {$\cdots$}
node[label, pos=0.9, anchor=center] {$\cdots$};
}

\edots{s}{p0}
\edots{s}{pl}

\edots{p0}{d0}
\edots{p0}{dj}
\edots{pi}{d0}
\edots{pi}{dm}
\edots{pl}{dj}
\edots{pl}{dm}

\edots{d0}{t}
\edots{dm}{t}

\end{tikzpicture}
\caption{
\label{fig:cm-flow}
Illustration of the flow network that is used to find a super-matching.
}
\end{figure}
One can verify that this is indeed a (integral) flow network and that there is a
straight forward translation between a flow and a super matching with the same weight.
	\subsection{3-approximation}
We now present a 3-approximation algorithm for the \textsc{\CARPOOL{}} problem.
This algorithm acts in two phases.
In the first phase it computes a maximum super-matching of $G$, 
in the second phase it decomposes the super-matching into 3 feasible
carpool matching and output the best of them.

We now describe how a super-matching can be decomposed into 3 feasible carpool
matching.
First, consider the graph obtained by an optimal super-matching.
Recall that in a super matching the out degree of every vertex is at most 1,
that is, the graph obtained by an optimal super matching is a pseudoforest -
every connected component has at most one cycle.
We now eliminate cycles from the super-matching by removing one edge from every
connected component.
It is easy to see that the resulting graph is a forest of anti-arborescences.
Each of these anti-arborescences can be, in turn, decomposed into two disjoint
feasible carpool matching.
This can be done, for example, by coloring each such anti-arborescences with two
colors, say red and green, and then consider the two solutions: one where the
green nodes are the drivers, and the other where the red nodes are the drivers.
We describe the algorithm in Algorithm~\ref{alg:cm3}, 
and illustrate it in Figure~\ref{fig:spanning-bipartite-graph}.   

\begin{algorithm}[t]
\KwIn{$G = (V, A), c : V \rightarrow \N, w : A \rightarrow \R$}											 
\KwOut{$(M \subseteq A)$}

$M \leftarrow \emptyset$								\\
$G' = (V, A') \leftarrow \text{superMatching($G$)}$				\\

\For{every connected component $C_i = (V_i, A_i) \in G'$}{	
	Eliminate the cycle in $C_i$ by removing an arc $a_i$							\\
	Decompose the remains anti-arborescences into two solutions, $M^i_1$, $M^i_2$	\\
	$M \gets M \cup \argmax_{F \in \{ \{e\}, M_1, M_2 \}}{w(F)}$	
}

\Return $M$
\caption{
\label{alg:cm3}
SuperMatching}
\end{algorithm}

\begin{figure}
\centering
\tikzset{
	span/.style={every node/.style={default node}, very thick},
	L/.style={red, dashed},
	R/.style={draw=blue},
}
\subfloat[][]{
\label{subfloat:graph}
\begin{tikzpicture}[span, ->]
\node(1) at (0,0) {1};
\node(2) at (2,0) {2};
\node(3) at (0,2) {3};
\node(4) at (2,2) {2};
\node(5) at (1,3.5) {1};

\draw(1) -- node[label, below]{4} (2);
\draw(1) -- node[label, above]{3} (3);
\draw(1) -- node[label, above]{5} (4);
\draw(2) -- node[label, above]{2} (4);
\draw(3) -- node[label, above]{2} (5);
\draw(4) -- node[label, above]{1} (3);
\draw(5) -- node[label, above]{4} (4);
\end{tikzpicture}}
\hfill\subfloat[][]{
\label{subfloat:super}
\begin{tikzpicture}[span, ->]
\node(1) at (0,0) {1};
\node(2) at (2,0) {2};
\node(3) at (0,2) {3};
\node(4) at (2,2) {2};
\node(5) at (1,3.5) {1};

\draw(1) -- node[label, below]{4} (2);
\draw(2) -- node[label, above]{2} (4);
\draw(3) -- node[label, above]{2} (5);
\draw(4) -- node[label, above]{1} (3);
\draw(5) -- node[label, above]{4} (4);
\end{tikzpicture}}
\hfill\subfloat[][]{
\label{subfloat:spanning}
\begin{tikzpicture}[span, ->]
\node[L](1) at (0,0) {1};
\node[R](2) at (2,0) {2};
\node[L](3) at (0,2) {3};
\node[L](4) at (2,2) {2};
\node[R](5) at (1,3.5) {1};

\draw(1) -- node[label, below]{4} (2);
\draw(2) -- node[label, above]{2} (4);
\draw(3) -- node[label, above]{2} (5);
\draw(5) -- node[label, above]{4} (4);

\end{tikzpicture}}
\hfill\subfloat[][]{
\label{subfloat:matching}
\begin{tikzpicture}[span, ->]
\node[L](1) at (0,0) {1};
\node[R](2) at (2,0) {2};
\node[L](3) at (0,2) {3};
\node[L](4) at (2,2) {2};
\node[R](5) at (1,3.5) {1};

\draw(1) -- node[label, below]{4} (2);
\draw(3) -- node[label, above]{2} (5);
\end{tikzpicture}}
\hfill
\caption[]{
\label{fig:spanning-bipartite-graph}
Illustration of the SuperMatching algorithm:
\subref{subfloat:graph} a directed graph. 
\subref{subfloat:super} a maximum super-matching.  
\subref{subfloat:spanning} an anti-arborescences:
$M_1$ is the set of arcs exiting red, dashed vertices, 
and $M_2$ is the set of arcs exiting blue vertices.
\subref{subfloat:matching} a feasible carpool matching with total value of 6.   
}
\end{figure}
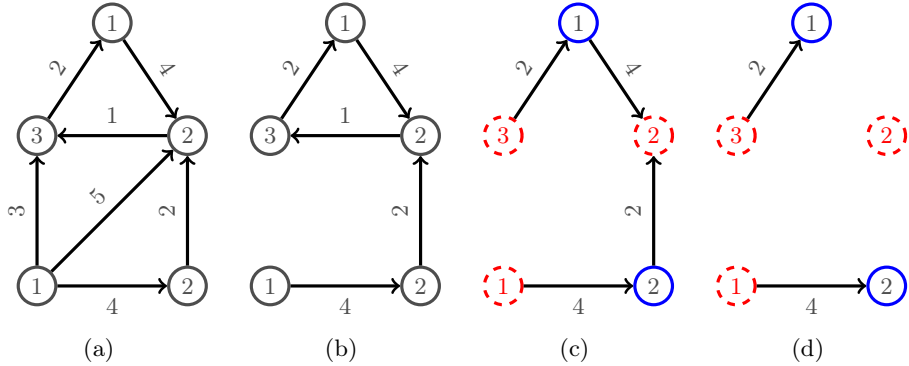

\begin{theorem}
Algorithm~\ref{alg:cm3} achieves a 3-approximation ratio.
\end{theorem}

\begin{proof}
Let $M_a = \bigcup_i \{a_i\}$ be the set of all removed arcs in the cycle
elimination phase.
Let $M_1 = \bigcup_i M^i_1$, and $M_2 = \bigcup_i M^i_2$.
Clearly, $M_a \cup M_1 \cup M_2 = A'$, 
and that $\max(w(M_a), w(M_1), w(M_2)) \geq
\frac{w(A')}{3}$.
The observation that the weight of a maximum super-matching is an upper bound on
the weight of a maximum carpool matching finishes the proof. 
\end{proof}

To see that our analysis is tight, consider the example in
Figure~\ref{fig:3cm-tight-fig}. 
Assume, for the given graph in the figure, 
that all weights are 1 and that there is no capacity constraint.
The maximum matching, then, is 3 ($\{(1,4), (2,4), (3,4)\}$), 
but the algorithm can return the super matching $\{(1,2), (2,3), (3,1)\}$ from which
only one arc can survive.  

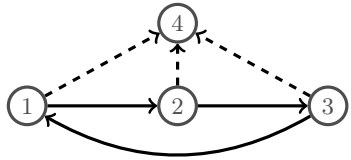
\begin{figure}
\centering
\begin{tikzpicture}

\node(n1)[default node] at (0,0) {1};
\node(n2)[default node] at (2,0) {2};
\node(n3)[default node] at (4,0) {3};
\node(n4)[default node] at (2,1.1) {4};

\draw[->, very thick] (n1) -- (n2);
\draw[->, very thick] (n2) -- (n3);
\draw[->, very thick] (n3) to[bend left] (n1);

\foreach \i in {1,...,3}
\draw[->, dashed, very thick] (n\i) -- (n4);

\end{tikzpicture}
\caption{
\label{fig:3cm-tight-fig}
Super Matching Algorithm, worst case example
}
\end{figure}

\bibliographystyle{plain}
\bibliography{main}

\begin{thebibliography}{10}

\bibitem{zimride}
Zimride by enterprise.
\newblock \url{https://zimride.com/}.

\bibitem{agatz2012optimization}
Niels Agatz, Alan Erera, Martin Savelsbergh, and Xing Wang.
\newblock Optimization for dynamic ride-sharing: A review.
\newblock {\em European Journal of Operational Research}, 223(2):295--303,
  2012.

\bibitem{arkin2004approximations}
Esther~M Arkin, Refael Hassin, Shlomi Rubinstein, and Maxim Sviridenko.
\newblock Approximations for maximum transportation with permutable supply
  vector and other capacitated star packing problems.
\newblock {\em Algorithmica}, 39(2):175--187, 2004.

\bibitem{chen2007improved}
Ning Chen, Roee Engelberg, C~Thach Nguyen, Prasad Raghavendra, Atri Rudra, and
  Gyanit Singh.
\newblock Improved approximation algorithms for the spanning star forest
  problem.
\newblock In {\em Approximation, Randomization, and Combinatorial Optimization.
  Algorithms and Techniques}, pages 44--58. Springer, 2007.

\bibitem{hartman2013optimal}
Irith Ben-Arroyo Hartman.
\newblock Optimal assignment for carpooling-draft.
\newblock {\em Draft}, 2013.

\bibitem{hartman2014theory}
Irith Ben-Arroyo Hartman, Daniel Keren, Abed~Abu Dbai, Elad Cohen, Luk Knapen,
  Davy Janssens, et~al.
\newblock Theory and practice in large carpooling problems.
\newblock {\em Procedia Computer Science}, 32:339--347, 2014.

\bibitem{knapen2013estimating}
Luk Knapen, Daniel Keren, Sungjin Cho, Tom Bellemans, Davy Janssens, Geert
  Wets, et~al.
\newblock Estimating scalability issues while finding an optimal assignment for
  carpooling.
\newblock {\em Procedia Computer Science}, 19:372--379, 2013.

\bibitem{knapen2014exploiting}
Luk Knapen, Ansar Yasar, Sungjin Cho, Daniel Keren, Abed~Abu Dbai, Tom
  Bellemans, Davy Janssens, Geert Wets, Assaf Schuster, Izchak Sharfman, et~al.
\newblock Exploiting graph-theoretic tools for matching in carpooling
  applications.
\newblock {\em Journal of Ambient Intelligence and Humanized Computing},
  5(3):393--407, 2014.

\bibitem{nguyen2008approximating}
C~Thach Nguyen, Jian Shen, Minmei Hou, Li~Sheng, Webb Miller, and Louxin Zhang.
\newblock Approximating the spanning star forest problem and its application to
  genomic sequence alignment.
\newblock {\em SIAM Journal on Computing}, 38(3):946--962, 2008.

\bibitem{tardos1985strongly}
{\'E}va Tardos.
\newblock A strongly polynomial minimum cost circulation algorithm.
\newblock {\em Combinatorica}, 5(3):247--255, 1985.

\end{thebibliography}

\end{document}